\documentclass[a4paper]{article}
\usepackage{amsmath}\allowdisplaybreaks[4]
\usepackage{amsthm}
\usepackage[ruled,vlined]{algorithm2e}
\usepackage{amsfonts}
\usepackage{amssymb}
\usepackage[sl]{caption}
\usepackage[top=3cm,bottom=3cm,left=2cm,right=2cm]{geometry}

\newcommand\figcaption{\def\@captype{figure}\caption}%
\newcommand\tabcaption{\def\@captype{table}\caption}%
\makeatother

\usepackage{fancyhdr}
\usepackage{color}
\usepackage{listings}
\newtheorem{theorem}{\hskip\parindent Theorem}
\newtheorem{lemma}{\hskip\parindent Lemma}
\newtheorem{remark}{\hskip\parindent Remark}
\newtheorem{definition}{\hskip\parindent Definition}
\newtheorem{corollary}{\hskip\parindent Corollary}

\newcounter{note}[section]

\renewcommand{\thenote}{\thesection.\arabic{note}}

\newcommand{\li}[1]{\refstepcounter{note}$\ll${\sf Li's
Comment~\thenote:} {\sf \textcolor{blue}{#1}}$\gg$}

\newcommand{\ignore}[1]{}

\begin{document}

\title{Uniform Information Exchange in Multi-channel Wireless Ad Hoc Networks}

\author{Li Ning \\ 
Center for High Performance Computing\\
Shenzhen Institutes of Advanced Technology, CAS\\
Shenzhen, P.R. China\\
li.ning@siat.ac.cn
\and 
Dongxiao Yu \\
Department of Computer Science\\
The University of Hong Kong\\
Hong Kong, P.R. China\\
dxyu@cs.hku.hk
\and
Yong Zhang\\
Center for High Performance Computing\\
Shenzhen Institutes of Advanced Technology, CAS\\
Shenzhen, P.R. China\\
zhangyong@siat.ac.cn
\and 
Yuexuan Wang\\
Department of Computer Science\\
The University of Hong Kong\\
Hong Kong, P.R. China\\
amywang@hku.hk
\and
Francis C.M. Lau\\
Department of Computer Science\\
The University of Hong Kong\\
Hong Kong, P.R. China\\
fcmlau@cs.hku.hk
\and
Shengzhong Feng\\
Center for High Performance Computing\\
Shenzhen Institutes of Advanced Technology, CAS\\
Shenzhen, P.R. China\\
sz.feng@siat.ac.cn
}

\date{\today}

\maketitle
\begin{abstract}
In the information exchange problem, 
$k$ packets that are initially maintained by $k$ nodes need
to be disseminated to the whole network as quickly as possible. 
We consider this problem in single-hop multi-channel networks of $n$ nodes, 
and propose a uniform protocol that with high probability accomplishes the dissemination 
in $O(k/\mathcal{F} + \mathcal{F}\cdot \log n)$ rounds, assuming $\mathcal{F}$ 
available channels and collision detection. This result is asymptotically optimal
when $k$ is large ($k \geq \mathcal{F}^2 \cdot \log n$). To our knowledge, this 
is the first uniform protocol for information exchange in 
multi-channel networks.
\end{abstract}

\section{Introduction}

In this paper, we study the information exchange problem
in a single-hop, multi-channel radio network.
There are $k$ nodes, called the \emph{source nodes}, in the network.
At the beginning, each of them holds a packet, and the
target is to disseminate these packets to the whole network as quickly
as possible.
Information exchange is one of the most fundamental operations that are
frequently called for in the smooth running of a network.


Using multiple channels obviously can greatly increase the throughput
of the network. A lot of works have been devoted 
to studying the utilization of multiple channels
in the derivation of faster communication protocols (e.g. \cite{CMK14,DKN12,DaumGGKN13,DGGN07,DGGN08,GilbertGKN09,GGNT12,HolzerLPW12,HolzerPSW11,ShiHYWL12,Wang2014}). 
All existing works however require that the network size $n$ 
be known a prior.
In ad hoc networks, knowing $n$ is usually a tough task, as it would consume
a large amount of time and energy for nodes to compute this global
parameter, and hence greatly increase the load of the network.
Additionally, in ad hoc networks, the network size could change
frequently due to nodes leaving and joining.
This consideration necessitates the design of uniform
protocols which do not require any prior information about network
parameters including the network size $n$ and the number of source
nodes $k$. 
Uniform protocols have better scalability and therefore
more suitable for implementation in reality.
%
In this paper, we propose a uniform protocol for information exchange
whose time complexity decreases linearly as the number of available
channels increases.

\subsection{Network Model and Problem Definition}
A multi-channel single-hop network is defined as follows. There are
$n$ nodes in the network, and any pair of which can communicate with
each other directly. But $n$ is not known to nodes. Time is divided into
synchronous rounds. There are $\mathcal{F}$ channels available in the
network. We use $1,\ldots,\mathcal{F}$ to denote these channels.
Even though these $\mathcal{F}$ channels are available to all the nodes, at
any time a node can select at most one channel to listen to or transmit
on. A node operating on a channel in a given round learns nothing about
events on the other channels. When a node $v$ listens to a channel, it
can receive a message if and only if there is only one node transmitting
on the channel. If two or more nodes transmit on the same channel, a
collision occurs and none of these transmissions would be successful.
We assume that nodes can detect collisions, i.e., nodes can distinguish
collision from silence. Furthermore, we consider the case of
non-constant $\mathcal{F}$ (larger than any constant), since otherwise,
using a constant number of channels will not break the $\Omega(k)$ lower
bound for information exchange that always holds in single-channel
networks.

The algorithm proposed in this paper is randomized, and hence the
analysis involves many random events.
We say that an event happens with high probability (with respect to $n$), if
it happens with probability $1 - 1/n^c$ for some constant $c > 0$.
%


The goal of
of information exchange is to disseminate some source nodes' packets
to the whole network, which is more precisely defined as follows.

\begin{definition}\label{def:info}
\emph{(Information Exchange.)}
In the information exchange problem, initially $k$ source
nodes are holding packets $\{P_1, P_2, \ldots, P_{k}\}$ respectively. 
It is required to disseminate all these $k$ packets to the whole network
as quickly as possible.
\end{definition}
Denote by $K$ the set of source nodes. Then $|K|=k$. We study the harsh
case of the information exchange problem where nodes have no idea about the
number of packets $k$ and the set of source nodes $K$. 
We asume that multiple packets can
be packed in a single message. 
It is easy to see that if
$k$ is small relative to the number of channels $\mathcal{F}$, the
benefit of multiple channels will be weakened, since in this case there
could be a single node selecting a channel such that
its transmission cannot be received by anyone. Thus, throughout
this work, we assume that $k\geq \mathcal{F}\log n$, which ensures that
when nodes uniformly select the channels, there are multiple nodes
operating on each channel with high probability. However, we must
point out that our algorithm can also solve the case where $k$ is small.



\subsection{Our Result and Technique}
In this paper, we give the first known uniform protocol for information
exchange in multi-channel networks. Our algorithm can disseminate all
$k$ packets to the whole network in $O(k/\mathcal{F} + \mathcal{F} \cdot
\log n)$ rounds with high probability when there are $\mathcal{F}$
available channels. When $k$ is large ($k\geq \mathcal{F}^2\log n$), our
algorithm shows a linear speedup considering the $\Omega(k)$ lower bound
for single-channel networks. Note that $\Omega(k/\mathcal{F})$ is a
trivial lower bound for information exchange with $\mathcal{F}$
available channels. Hence, our protocol is asymptotically optimal when
$k$ is large.



In our protocol, every node that needs to transmit maintains a transmission probability,
and in each round, a node decides to transmit with its transmission probability. 
With $\mathcal{F}$ available channels, our protocol applies a very
intuitive rule for the nodes to do the selection: in each round, a node
just selects one channel uniformly at random, and then transmits or listens on the selected
channel. If a node listens on the channel and detected that the selected channel is \emph{idle}, then it doubles its transmission probability. Otherwise, the node halves the transmission probability.

By a straightforward computation, it is easy to discover that in order
to ensure a successful transmission on one channel with constant
probability,
the total transmission probability of nodes selecting this channel
should be a constant. Hence, to efficiently make use of the channels,
the total transmission probability of nodes should be in a ``safe
range'' $[\alpha_1\cdot \mathcal{F}, \alpha_2\cdot \mathcal{F}]$ with
constants $\alpha_1, \alpha_2 > 0$.
However, this is not easy to achieve without a careful design of the protocol.
The difficulty comes from the selection process of the channels. Since the nodes 
select the channels using a distributed, randomized protocol and the selections of nodes 
are mutually independent, the total transmission probability of nodes selecting a 
channel may vary a lot among different channels. As a result, nodes may
update their transmission probabilities towards different directions,
which makes it very hard to analyze whether the ``safe range'' is still
guaranteed after an update. Our protocol has a
channel-consistent-updating property, i.e., nodes selecting the same
channel update their transmission probabilities consistently. With this
channel-consistent-updating property, our analysis shows that
when the total transmission probability of all the nodes goes outside
the ``safe range'',
there are enough channels on which the nodes behave consistently to
pull the total transmission probability back
to the safe range. Then as a global effect, the network stays stable in the ``safe range'' state,
and the $\mathcal{F}$ available channels are used efficiently (i.e.\ transmit successfully
$\Omega(\mathcal{F})$ messages with constant probability in each round).

Notice that if we can quickly aggregate all the $k$ packets to a single node and 
let this node be the only one to transmit in the network (by broadcasting 
in a pre-defined primary channel), then the whole network will know all the $k$ packets very soon. 
Our protocol follows this approach, and hence prevents the node from trying to transmit anymore
if its message has been successfully received by some node which is still transmitting. 
This technique is also known as \emph{indirecting}, and can be summarized as ``if your
message is received by another speaker, then you never speak again''. By this approach, 
the nodes that try to transmit become fewer and fewer as the protocol is
running.

When the number of transmitting nodes becomes small, using all $\mathcal{F} > 1$ channels
does not always help. The reason is when there are only a few nodes to transmit, it is hard for them 
to meet each other if they still select the channel randomly. A direct solution 
will be such that if the transmitting nodes find that there are only a few left nodes, 
then they stop selecting channels and operate on a pre-defined channel.
The primary channel that is designed for the final broadcast can be used to achieve this.
However, the problem has not been completely solved.
By our analysis, we know that 
the number of rounds needed before the multiple channels become inefficient is 
$T = O(\log n + k/\mathcal{F})$ with high probability. Unfortunately,
this time bound cannot be known
to the nodes in the network, since the protocol is uniform and
information about $n$ and $k$ is not
available to the nodes. Consequently, the nodes cannot calculate $T$
and have no idea about
when to stop selecting channels. Our algorithm uses a parallel approach to overcome this difficulty. 
In our protocol, each round is divided into four slots: in the first two
slots, nodes use multiple channels for transmissions; 
and in the last two slots, all nodes only use the pre-defined primary
channel for transmission.
This parallel approach affects the running time by only a constant factor, but perfectly solves the inefficiency problem of multiple channels for information exchange with small number of transmitting nodes. Our analysis shows that such a parallel approach completes the information exchange in 
$O(k/\mathcal{F} + \mathcal{F} \cdot \log n)$ rounds with high probability.


\subsection{Related Work}
As more and more wireless networks and devices now operate on multiple
channels, there has been much attention given to studying the effect of
multiple channels on facilitating communication recently
\cite{CMK14,DKN12,DaumGGKN13,DGGN07,DGGN08,GilbertGKN09,GGNT12,HolzerLPW12,HolzerPSW11,ShiHYWL12,Wang2014}.
With respect to information exchange in multi-channel single-hop
networks, most studies are done under the assumption that each message
can carry only one packet. In particular, Holzer et al.\
\cite{HolzerPSW11,HolzerLPW12} proposed deterministic and randomized
algorithms with optimal $O(k)$ time to solve the information exchange
problem. With the assumption that nodes can listen to and
receive messages from multiple channels at the same time, Shi et al.\ \cite{ShiHYWL12} proposed
an $O(\log k \log \log k)$ time randomized information exchange protocol
using $\Theta(n)$ channels. But with the assumption of unit-size messages,
the benefit of utilizing multiple channels is very limited, since in
each round, a node can receive at most one packet.
Hence, it needs $\Omega(k)$ rounds to complete the information
exchange. On the other hand, the packet stored at nodes could be small
(e.g., in sensor networks, the data at each node is only a value). It is
realistic to consider the case that multiple packets can be packed in a
single message. Under this assumption, in \cite{DaumGGKN13}, Daum et al.
proposed a randomized algorithm that accomplishes
information exchange in $O(k + \log^2 n/\mathcal{F} +\log n\log \log n)$ rounds with high probability. Their algorithm does not rely on collision detection. Then with collision detection, Wang et al.\ \cite{Wang2014} proposed a protocol that 
disseminates all the packets in $O(k/\mathcal{F} + \mathcal{F} \cdot \log^2 n)$ rounds
with high probability. When $k$ is large ($k\geq \mathcal{F}^2\log^2n$), this result is asymptotically optimal considering the trivial lower bound $\Omega(k/\mathcal{F})$. In~\cite{YYWYL14}, Yan et al. studied the impact of message size on information exchange in multi-channel networks. Additionally, Gilbert et al.\ \cite{GilbertGKN09} considered the scenario when an adversary can disrupt a number of channels 
and proposed a randomized algorithm to achieve the almost-complete
information exchange. However, all the above results need the prior knowledge
of $n$. To our knowledge, there is not yet any uniform protocol proposed
for solving the information exchange problem
in single-hop multiple-channel networks.

Information exchange has also been extensively studied
since 1970s \cite{Capetanakis79,Hayes1978,Mikhailov2005} in single-channel networks. In single-channel networks, information exchange is also known
as contention resolution \cite{Fernandez2013} or $k$-selection \cite{Kowalski2005}. 
Assuming collision detection as in this work, 
a randomized adaptive protocol with expected running time of $O(k + \log n)$ was
presented by Martel in \cite{Martel1994}. Kowalski \cite{Kowalski2005} improved the protocol in \cite{Martel1994} to
$O(k +\log \log n)$ by making use of the expected $O(\log \log n)$ selection
protocol in \cite{Willard1986}. When requiring high probability results, the best known 
randomized algorithm was introduced in \cite{AntaM10}, which solves the $k$-selection 
problem in $O(k +\log^2 n)$ rounds without assuming collision detection. Note that
in the single-channel networks, the trivial lower bound for $k$-selection is $\Omega(k)$. 
Hence the result in \cite{AntaM10} is asymptotically optimal for $k\in \Omega(\log^2 n)$. By assuming that the channel can provide feedback on whether a message is successfully transmitted,
an uniform randomized protocol with running time $O(k)$ is introduced in \cite{Fernandez2013} for single-channel networks.
However, the error probability of the protocol in \cite{Fernandez2013} is $1/k^c$, rather than $1/n^c$. 
For deterministic solutions, adaptive protocols for $k$-selection were presented with running 
time $O(k \log (n/k))$ in \cite{Capetanakis79,Hayes1978,Mikhailov2005}, assuming collision detection. 

%
%
%

\subsection{Outline}
Section~\ref{sec:pre} introduces some preliminary results that help
the analysis. Section~\ref{sec:protocol} introduces our protocol. 
Section~\ref{sec:analysis} analyzes the performance of our protocol;
particularly, we give an upper 
bound on the time needed to accomplish (with high probability)
information exchange.
Furthermore, we show the ``self-stabilization'' property of our protocol.
Section~\ref{sec:con} summarizes our work, followed by a discussion.  

\section{Preliminaries}\label{sec:pre}
\ignore{
Recall that the number of available channels is denoted by $\mathcal{F}$. We assume that $\mathcal{F}$ can be
very large, since without this assumption using more channels does not decrease the time complexity
bound for the algorithms solving the information exchange problem. To see why, 
consider the case when there is a constant $\Gamma > 0$ such that $\mathcal{F}$ is always smaller than $\Gamma > 0$.
Then any protocol that uses $\mathcal{F}$ channels needs $k/\mathcal{F} = \Omega(k)$ time rounds
to solve the information exchange problem, which is at the same time complexity order with the lower bound
for the case when only one channel is available.

Furthermore, another practical assumption should be also noticed. That is, even though $\mathcal{F}$ can be 
very large (not bounded by any constant), it always hold that the network size $n$ is still much bigger
than $\mathcal{F}$.
}
In this section, we review some useful results concerning randomness.

\begin{lemma}[Chernoff Bound.]\label{lem:chernoff}
Consider a set of random variables $0 \leq X_1$, $X_2$, \ldots, $X_n \leq c$ for some parameter $c > 0$. 
Let $X := \sum^n_{i=1} X_i$ 
and $\mu := \mathbb{E}[X]$. If $X_i$'s are independent or negatively associated, 
then for any $\delta > 0$ it holds that 
\[ Pr[X \geq (1 + \delta)\mu] \leq \left(\frac{e^\delta}{(1 + \delta)^{1 + \delta}}\right)^{\frac{\mu}{c}}.\]
In details, for $\delta \leq 1$, the bound can be upper bounded by 
\[Pr[X \geq (1 + \delta)\mu] \leq \exp\{-\frac{\delta^2\mu}{3c}\};\]
for $\delta > 1$, it holds that 
\[Pr[X \geq (1+\delta)\mu] \leq \exp\{-\frac{\delta\ln(1+\delta)\mu}{2c}\}.\] 
On the other hand, for any $0 < \delta < 1$ it holds that 
\[Pr[X \leq (1 - \delta)\mu] \leq 
\left(\frac{e^{-\delta}}{(1-\delta)^{1-\delta}}\right)^{\frac{\mu}{c}} \leq \exp\{-\frac{\delta^2\mu}{c}\}.\]
\end{lemma}

Next, we present some useful conclusions about a classic procedure, ``throw balls into bins''.
These conclusions have essentially been proved in existing works such as \cite{Daum2012a}. However, 
for the completeness of our arguments, we will go through the proof in
details. 

\begin{lemma}\label{lem:bin_ball}
Consider $H$ bins and $l$ balls with weights $0 \leq w_1$, $w_2$, \ldots, $w_l \leq \zeta$.
Assume that $\sum_{i=1}^l w_i = \alpha \cdot H$ where $\alpha \geq 0.01$ is a constant. 
Balls are thrown into bins uniformly at random. Then, if $\zeta$ is small enough, 
with probability $1 - \exp\{-\Omega(H)\}$
there are at least $H \cdot 31/32$ bins in which the total weight of balls is between $\alpha \cdot 15/16$ 
and $\alpha \cdot 2$.
\end{lemma}
\begin{proof}
Considering the $j$-th bin, let $X^i_j$ denote the random variable that
takes value $w_i$
if the $i$-th ball is in the $j$-th bin, and $0$ otherwise. 
Since the balls are thrown to bins uniformly at random, $Pr[X^i_j = w_i] = 1/H$ and 
$\mathbb{E}[\sum_i X^i_j] = \alpha$.

Let $Y_j$ denote the binary random variable that takes value
$1$ if the total weight of balls thrown into the $j$-th bin is at least $\alpha \cdot 15/16$. 
Consequently, by Chernoff bound we have
\begin{equation}
\begin{aligned}
\mathbb{E}[Y_j] = Pr[Y_j = 1]&=1-Pr[\sum_i X^i_j < \alpha \cdot \frac{15}{16}] \\
&\geq1-\exp\{- \frac{\alpha}{16^2\cdot2\cdot \zeta}\},
\end{aligned}
\end{equation}
which implies $\mathbb{E}[Y_j] > 127/128$ when $\zeta$ is small enough to 
promise that $\exp\{- 0.01 / \left(16^2\cdot2\cdot \zeta\right)\} < 1/128$.

Let $Z_j$ denote the binary random variable that takes value
$1$ if the total weight of balls thrown into the $j$-th bin is at most $\alpha \cdot 2$. 
Consequently, by Chernoff bound we have
\begin{equation}
\begin{aligned}
\mathbb{E}[Z_j] = Pr[Z_j = 1]&=1-Pr[\sum_i X^i_j > \alpha \cdot 2] \\
&\geq1-\exp\{- \frac{\alpha}{3\zeta}\},
\end{aligned}
\end{equation}
which implies $\mathbb{E}[Z_j] > 127/128$ when $\zeta$ is small enough to
promise that $\exp\{- 0.01 / \left(3\zeta\right)\} < 1/128$.

Hence, we have $\mathbb{E}[\sum_j Y_j] > H \cdot 127/128$ and $\mathbb{E}[\sum_j Z_j] > H \cdot 127/128$.
Note that $Y_1, \ldots, Y_H$ are negatively associated, as well as $Z_1, \ldots, Z_H$ \cite{Daum2012a}.
Hence, by Chernoff bound, it holds that with probability $1 - \exp\{-\Omega(H)\}$ there are
at least $H \cdot 31/32$ bins with weight between $\alpha \cdot 15/16$ and $\alpha \cdot 2$.
\end{proof}

\begin{lemma}\label{lem:bin_ball_basic}
Consider $H$ bins and $ l > H \cdot \Delta$ balls, where $\Delta > 2$. 
Balls are thrown to bins uniformly at random. 
If $\Delta$ is big enough, then with probability $1 - \exp\{-\Omega(H)\}$ 
there are at least $H \cdot 31/32$ bins that contain at least $2$ balls.
\end{lemma}
\begin{proof}
Considering the $j$-th bin, let $X^i_j$ denote the random variable that takes value $1$
if the $i$-th ball is in the $j$-th bin, and $0$ otherwise. 
Since the balls are thrown to bins uniformly at random, then $Pr[X^i_j = 1] = 1/F$ 
and $\mathbb{E}[\sum_i X^i_j] > \Delta$.

Let $Y_j$ denote the binary random variable that takes value
$1$ if the number of balls thrown into the $j$-th bin is at least $2$. 
Consequently, by Chernoff bound we have
\begin{equation}
\begin{aligned}
\mathbb{E}[Y_j] = Pr[Y_j = 1] &=1-Pr[\sum_i X^i_j < 2]  \\
&=1-Pr[\sum_i X^i_j < \frac{2}{\Delta} \cdot \Delta] \\
&\geq1-\exp\{- \left(\frac{\Delta - 2}{\Delta}\right)^2 \cdot \frac{\Delta}{2}\},
\end{aligned}
\end{equation}
which implies $\mathbb{E}[Y_j] > 63/64$ when $\Delta$ is big enough.
By Chernoff bound, it holds that with probability $1 - \exp\{-\Omega(H)\}$ there are
at least $H\cdot 31/32$ bins in which there are at least $2$ balls.
\end{proof}

\begin{corollary}\label{cor:bin_ball}
Consider $H$ bins and $l > H \cdot \Delta$ balls with weights $0 \leq w_1$, $w_2$, \ldots, $w_l \leq \zeta$.
Assume that $\sum_{i=1}^l w_i = \alpha \cdot H$ where $\alpha \geq 0.01$ is a constant. 
Balls are thrown to bins uniformly at random. Then, if $\Delta$ is big enough and $\zeta$ is small enough, 
then with probability $1 - \exp\{-\Omega(H)\}$ there are at least $H \cdot 15/16$ bins 
in which there are at least $2$ balls, and the total weight is between $\alpha \cdot 15/16$ and $\alpha \cdot 2$.
\end{corollary}
\begin{proof}
The conclusion is implied directly from Lemma~\ref{lem:bin_ball} and Lemma~\ref{lem:bin_ball_basic} .
\end{proof}

At the end of this section, we introduce a result given in \cite{AwerbuchRS08}.

\begin{lemma}\label{lem:q1_q0}
Consider a set of $l$ nodes, $v_1, v_2, \ldots, v_l$, transmitting on a channel. 
For node $v_i$, it transmits with probability $0 < p(v_i) < 1/2$. 
Let $w_0$ denote the probability that the channel is idle; and $w_1$ 
the probability that there is exactly one transmission
on the channel. Then, $w_0 \cdot \sum_{i=1}^l p(v_i) \leq w_1 \leq 2 \cdot w_0\cdot \sum_{i = 1}^l p(v_i)$. 
\end{lemma}

The proof is omitted; readers can refer to \cite{AwerbuchRS08}
for the detailed proof.

\section{Uniform Information Echange} \label{sec:protocol}

In this section, we introduce our Uniform Information Exchange (UIE)
protocol. The pseudo-code of the protocol is given in
Algorithm~\ref{algo:1} and Algorithm~\ref{algo:2}.

\textbf{UIE Protocol.}  There are two states for the nodes: \emph{active} and \emph{inactive}. 
Intuitively, the active nodes are trying to transmit messages over the network, while the inactive
nodes just listen for incoming messages. Initially, all the source nodes
are \emph{active}, and the others are \emph{inactive}.

In the protocol, an active node will become inactive when it successfully
transmits its message to other active nodes. This way, on
one hand, the number of active nodes is constantly decreasing, and on the
other hand, it ensures that at any time the active nodes possess all $k$
packets. Hence, when there is only one active node left, it can send all
the $k$ packets to all the other nodes. The utilization of multiple
channels can speed up the reduction of active nodes. By the
transmissions on multiple channels, the active nodes can be reduced on
all channels in parallel. However, as discussed before, when the number
of active nodes becomes small, it cannot guarantee that for a
particular channel, there are multiple active nodes operating on it. As
a result, even if an active node successfully transmits on a channel, its
message may not be received by other active nodes. In other words, the
multiple channels are not efficient any more. Additionally, the protocol
needs to ensure that when the surviving active node transmits, all other
nodes listen on the same channel. Hence, we set a primary channel, which
serves two purposes: first, it is used for
reducing active nodes when the number of active nodes is small;
second, it is used by the surviving active node to disseminate the packets.

Specifically, there are two processes in the protocol: the
multiple-channel transmission process and the primary-channel
transmission process. In the multiple-channel transmission process,
active nodes operate on multiple channels to reduce the number of active
nodes, while in the primary-channel transmission process, nodes operate
on the primary channel. Note that because nodes have no idea about any
network parameters, it is hard for nodes to determine when the
multiple-channel transmission process should finish. Hence, in the
protocol, these two processes are in parallel, rather than consecutive.
Specifically, there are four slots in each round: in the first two slots,
active nodes operate on multiple channels, and in the other slots, nodes
operate on the primary channel. We set the first channel as the special
primary channel. We next introduce the protocol in more detail.

Each active node $v$ maintains two parameters $p(v)$ and $q(v)$.
Denote the values of $p(v)$ and $q(v)$ in a round $t$ by $p_t(v)$ and
$q_t(v)$, respectively. In particular, $p_t(v)$ and $q_t(v)$ are the
transmission probabilities of node $v$ for the multi-channel
transmission process and the primary-channel transmission process in
round $t$, respectively. Initially, $p_0(v) :=q_0(v) := \zeta$, where $0
< \zeta < 1$ is a constant (determined in Lemma~\ref{lem:bin_ball}).
Let $m_t(v)$ denote the set of packets received by node $v$ by round $t$. Initially, 
for a source node $v$ initiated with packet $P$, $m_0(v) := \{P\}$. And for other nodes, $m_0(v) := \emptyset$.



%

The operations in the four slots of each round $t$ are as follows:
\begin{itemize}
\item 
\textbf{Slot 1.} In this slot, the inactive nodes do nothing.
Each active node $v$ 
selects a channel from the $\mathcal{F}$ candidates uniformly at random,
and then transmits
with probability $p_t(v)$ on the selected channel. If it does not transmit, it listens
on the selected channel. If $v$ receives a message containing a set of
packets $m'$, it updates $m_{t+1}(v) := m' \cup m_t(v)$.

At the end of Slot 1, $v$ updates the transmission probability $p$
according to the following rule: if $v$ listens and detects no
transmission on the selected channel, $p_{t+1}(v):= \min\{\zeta, 2\cdot
p_t(v)\}$; otherwise, $p_{t+1}(v):= p_t(v) / 2$.
\item 
\textbf{Slot 2.} In this slot, the inactive nodes still do nothing. For
an active node $v$, if it has received a message in Slot 1, it transmits an
acknowledgement on the selected channel. Otherwise, $v$ listens on the
selected channel.

If an active node $v$ transmitted in slot 1 and detects transmissions on
the selected channel in Slot 2, the state of $v$ switches to
\emph{inactive}.


\item 
\textbf{Slot 3.} In this slot, all nodes operate on the primary channel
(Channel 1). Specifically, all inactive nodes listen, and an active
node $v$ transmits with probability $q_t(v)$. At the end of Slot 3,
active nodes update the transmission probability $q_t(\cdot)$ using the
same rule as in Slot 1.

\item 
\textbf{Slot 4.} For each (active or inactive) node $v$, if $v$ received a message in Slot 3, it transmits an acknowledgement.

For an active node $v$, if $v$ transmitted in Slot 3 and detects transmissions in this slot, it changes its state to \emph{inactive}

%
\end{itemize}


\begin{algorithm}
\textbf{Initialization: for node $v$ at time $0$} \\ 
\nl	$p(v) := q(v) := \zeta$; \\
\nl	\If{initially have a packet $P$}{ 
\nl		$m(v) := \{P\}$; \\
\nl		$state(v) := Active$; \\
	}
\nl	\Else{  
\nl		$m(v) :=\{\}$; \\
\nl		$state(v) := Inactive$; \\
	}
\vspace{0.2cm}
\textbf{Active State: for node $v$ at time $t \geq 0$}\\
\nl \textbf{Slot 1-2:} pick a Channel $r$ uniformly at random and call $channel-use(r, m(v),p(v))$;\\
\nl \textbf{Slot 3-4:} call $channel-use(1, m(v), q(v))$;\\
\vspace{0.2cm}
\textbf{Inactive State: for node $v$ at time $t \geq 0$}\\
\nl \textbf{Slot 1-2:} do nothing;\\
\nl \textbf{Slot 3:} \\
\nl	listen on Channel $1$;\\
\nl	\lIf{receive a message containing a set of packets $m'$}{$m(v) := m(v) \cup m'$;}\\
\nl \textbf{Slot 4:} \lIf{received a message in Slot $3$}{transmit on Channel $1$;}\\
\caption{UIE}
\label{algo:1}
\end{algorithm}

\begin{algorithm}
\textbf{Slot 1:} \\
\nl	on Channel $i$, transmit a message containing packets in $s$ with probability $w$ and listen with probability $1-w$; \\
\nl \If{listened}{
\nl		\If{Channel $i$ is idle} {
\nl			$w := \min\{2w, \zeta\}$;\\
		}
\nl		\ElseIf{received a message containing a set of packets $m'$} {
\nl		 	$s := s \cup m'$; \\
\nl			$w := w/2$;\\
		}
\nl		\Else{ 
\nl			\textrm{// Channel $i$ is busy} \\
\nl			$w := w/2$;\\
		}
	}
\nl	\Else{
\nl		\textrm{// transmitted} \\
\nl		$w := w/2$;\\
}
\vspace{0.3cm}
\textbf{Slot 2:} \\
\nl	\lIf{received a messge in Slot $1$}{transmit on Channel $i$;}\\
\nl	\If{transmited in Slot $1$}{
\nl		listen on Channel $i$; \\
\nl		\lIf{receive a message OR Channel $i$ is busy}{$state(v) := Inactive$;}\\
	}
	
\caption{$channel-use(i,s,w)$}
\label{algo:2}
\end{algorithm}

We state the correctness of the UIE protocol in the following Theorem~\ref{th:co}.
\begin{theorem}\label{th:co}
Consider an execution of the UIE Protocol. When there is exactly one active node left, 
say node $v$ in round $T$, then $p_{T}(v) = \bigcup_{v\in K} m_0(v)$. 
Recall that $K$ is the set of all source nodes.
\end{theorem}
\begin{proof}
Denote the set of active nodes in round $t$ by $A_t$. 
Then $A_t \subseteq A_{t-1} \subseteq \ldots \subseteq A_1 \subseteq A_0 = K$ holds for any $t > 0$, according
to the protocol. Then the conclusion follows from the fact that 
$\bigcup_{v\in A_t} m_t(v) = \bigcup_{v\in A_{t-1}} m_{t-1}(v)$ holds for any $t > 0$, which 
is true because when an active node $v$ becomes inactive in some round $t$, it means $m_t(v)$ is known
to some other active node $u$ which is still active in round $t+1$. In detail, if an active node received
acknowledgement or detected collisions in Slot $2$, then it means its message has been received by some other active nodes; 
if an active node received an acknowledgement or detected collisions in Slot 4, then it means its message has 
been received by all the other nodes in the network, including the active ones if any exists.
\end{proof}

\ignore{
\begin{remark}
Suppose there is one active node $v$ that is
left with $p_{T}(v) = \bigcup_{v\in K} m_0(v)$, where 
$T = \Omega(\log n)$ (as shown later in Section~\ref{sec:analysis}). Then it is easy 
to find out that after another $O(T)$ rounds, $\bigcup_{v\in K} m_0(v)$ will be known to all 
the nodes with high probability. Note that all the nodes other than node $v$ are inactive after 
time $T$. To see why this it true, note that node $v$ needs at most $T$ rounds to recover
its transmit probability $q(v)$ to $\zeta$. Hence, with high probability, during the $T$ rounds thereafter,
node $v$ can disseminate the message containing all $k$ packets to all other nodes on the primary channel.
\end{remark}
}

\section{Analysis of the Protocol}\label{sec:analysis} 

In this section, we prove that with $k$ source nodes, our protocol can disseminate all $k$ packets to the whole network in $O(k/\mathcal{F} + \mathcal{F} \cdot \log n)$ rounds with high probability.
Recall that $\mathcal{F}$ is the number 
of available channels and $n$ is the number of nodes in the network. Formally, this conclusion is summarized in Theorem~\ref{thm:info}.

\begin{theorem}\label{thm:info}
Consider information exchange on a network of size $n$ with
$\mathcal{F}$ available
channels. For the case where there are initially $k \leq n$ source
nodes, the following conclusions hold: 
\begin{enumerate}
\item
There exists a constant $\nu > 0$ such that with high probability, there is
only one active node left at time $T^* := \nu(k/\mathcal{F} +
\mathcal{F} \cdot \log n)$.
\item
For time $T^*$ when there is only one active node $v$ left,
at time $T^{**} := 2 \cdot T^* + \log n= O(T^*)$ it holds 
with high probability that every node in the network knows the $k$
packets initially maintained
by the source nodes and node $v$ becomes inactive.
\end{enumerate}
\end{theorem}
\begin{proof}
The first conclusion follows directly from the
Lemmas~\ref{lem:multi} and \ref{lem:primary} which will be given later in
Section~\ref{sec:multi} and \ref{sec:primary}, respectively.
Here we first prove the second conclusion.

Since at time $T^*$ there is only one active node $v$ left, we know that $q_{2T^*}(v)$ will get 
back to $\zeta$ if node $v$ is still active at time $2T^*$. Note that 
if node $v$ transmits with probability $\zeta$ on the primary channel (Slot $3$) for the subsequent $\log n$ rounds, 
then with high probability there exists one round in which node $v$ transmits and consequently
all the nodes in the network will receive the message. As shown in Theorem~\ref{th:co}, the message transmitted by $v$ contains all $k$ packets. Hence, all nodes will get these packets in the received message. Finally, since the inactive nodes that received 
a message in Slot $3$ transmit on the primary channel, then node $v$ detects transmissions 
in Slot $4$ and becomes inactive.
\end{proof}

We next briefly introduce the analysis process for the first conclusion
in Theorem~\ref{thm:info}. Recall that there are two parallel processes
in our algorithm: the multi-channel transmission process (the first two
slots in each round) and the primary-channel transmission process (the
last two slots in each round). As discussed before, when there are many
active nodes (more than $\mathcal{F} \cdot \log n$), multiple channels
should be efficient in reducing the number of active nodes. When the
number of active nodes is reduced to something small (less than $\mathcal{F}
\cdot \log n$), the utilization of multiple channels might not be
efficient any more, since for a particular channel, there might not be
multiple nodes selecting it. In this case, we have to rely on the
primary-channel transmission process to reduce the number of active
nodes. Therefore, we divide the analysis into two parts. The first part
analyzes how long it takes to decrease the number of active nodes to
$\mathcal{F} \cdot \log n$ and the second part deals with how long it
takes to further reduce the number of active nodes to one. More
precisely, let $A_t$ denote the set of active nodes in a round $t$. Let
$T$ be the first round in which the number of
active nodes drops below $\mathcal{F} \cdot \log n$. That is, for any time $t < T$, it holds that
$|A_t| \geq \mathcal{F}\cdot \log n$, and for any time $t \geq T$, it
holds that $|A_t| < \mathcal{F} \cdot \log n$. Then the whole analysis
is divided into two parts by $T$: The first part concerns the
time period from
$0$ to $T-1$, and the second part considers the algorithm execution
since $T$. In the first part of the analysis, we mainly analyze the
efficiency of the multi-channel transmission process in reducing the
number of active nodes, and in the second part, we are mainly concerned
about the efficiency of the primary-channel transmission process.

In the rest of this section, we assume that $k \geq \mathcal{F} \cdot
\log n$. Otherwise, we can jump directly to the second part of the analysis.


\subsection{Efficiency of Multiple Channels}\label{sec:multi}

In this section, we analyze the first part, i.e., the period from time
$0$ to the first round when the number of active nodes drops below
$\mathcal{F}\cdot \log n$. The conclusion is summarized in the following
Lemma~\ref{lem:multi}.

\begin{lemma}\label{lem:multi}
There exists $T = O(k / \mathcal{F})$ such that in round $T$ it holds
with high probability that
$|A_T| < \mathcal{F} \cdot \log n$.
\end{lemma}

The main idea in proving Lemma~\ref{lem:multi} is to find a proper
$\gamma' > 0$ such that after the
protocol has been running for $T' = O(\log n)$ rounds, within any period of 
$\gamma'$ rounds subsequently with $|A_t| \geq \mathcal{F} \cdot \log n$, there are (with constant probability)
$\Omega(\mathcal{F})$ active nodes that switch from the \emph{active} state to the \emph{inactive} state.
Then, with high probability, there are $k - \mathcal{F} \cdot \log n < k$ active nodes switching to \emph{inactive},
in a period of $O(\log n + k/\mathcal{F})$ (which is $O(k/\mathcal{F})$
for $k > \mathcal{F}\cdot \log n$) rounds. To prove
Lemma~\ref{lem:multi}, we need to introduce and prove a series of
``small'' lemmas at first, and leave the proof of
Lemma~\ref{lem:multi} to the end of this section. We next do some
preparation for proving Lemma~\ref{lem:multi}.

By using more channels, it is natural to expect that the number of 
successful transmissions is increased accordingly. 
Specifically, it is expected that in a round, 
there should be $\Omega(\mathcal{F})$ successful transmissions with
$\mathcal{F}$ channels. In the following Lemma~\ref{lem:safe}, we show
that if a ``safe range'' on the total transmission probability of all
active nodes is satisfied, the above expectation is true.

\begin{lemma}[Safe range]\label{lem:safe}
Consider the Uniform Information Exchange Protocol. For a round $t > 0$ with $|A_t| \geq \mathcal{F} \cdot \log n$, 
if there exist constants $\alpha_1, \alpha_2 \geq 1$ such that 
$\alpha_1 \cdot \mathcal{F} \leq \sum_{v \in A_t} p_t(v) \leq \alpha_2 \cdot \mathcal{F}$, 
then with constant probability there are $\Omega(\mathcal{F})$ active nodes switching to the inactive state 
in the second slot.
\end{lemma}
\begin{proof}
For the convenience of the argument, we introduce a series of random variables 
$X^i(v)$ with $i = 1, \cdots, \mathcal{F}$ and $v \in A_t$. The variable $X^i(v)$ takes value 
$p_t(v)$ if node $v$ selects Channel $i$ in the $1$st slot of round $t$; 
otherwise, $X^i(v) := 0$. Furthermore, denote $X^i := \sum_v X^i(v)$.
By Corollary~\ref{cor:bin_ball}, with probability $1 - \exp\{-\Omega(\mathcal{F})\}$, 
there are at least $\mathcal{F} \cdot 15/16$ channels, such that
for each of them there are at least two active nodes selecting it and
the total transmission probability of these active nodes is between
$\alpha_1 \cdot 15/16$ and $2 \cdot \alpha_2$.
Next, we show that in such cases there are $\Omega(\mathcal{F})$ active
nodes switching to the \emph{inactive} state with constant probability.

With $b_i \in [0,1/2]$ for $i = 0, 1, \ldots$, it holds \cite{DaumGGKN13} that 
\begin{eqnarray}\label{eqn:idle}
4^{-\sum_i b_i} \leq \prod_i (1-b_i) \leq e^{-\sum_i b_i}.
\end{eqnarray}
Hence, for Channel $i$ with $X^i$ between $\alpha_1 \cdot 15/16$ and $2\cdot \alpha_2$, it is idle with 
probability at least $4^{-4\alpha_2}$, and there is exactly one transmission on the channel with probability at least 
$\alpha_1 \cdot 4^{-4\alpha_2} \cdot 15/16$ (by Lemma~\ref{lem:q1_q0}). 
If there are at least two active nodes
selecting a channel and there is only one node transmitting on the
channel, then the transmission will succeed and the one transmitting in
the first slot will sense transmissions in the second slot. According to
the algorithm, the node that transmitted will switch to the
\emph{inactive} state. Therefore there are at least $\mathcal{F} \cdot 15/16$ channels such that for each of them
there is an active node switching to the \emph{inactive} state with probability at least $\alpha_1 \cdot 4^{-4\alpha_2} \cdot 15/16$. 
In expectation, there are $C\cdot \mathcal{F}$ new \emph{inactive} nodes where  
$C := (1-\exp\{-\Omega(\mathcal{F})\})\alpha_1 \cdot 4^{-2\alpha_2} \cdot 15/16$ which is 
at least $\alpha_1 \cdot 4^{-2\alpha_2} \cdot 15/32$ when $\mathcal{F}$
is large enough. Using the Chernoff bound
over the $\mathcal{F}$ channels, it holds that with constant probability (given $\alpha_1$, $\alpha_2$, and large enough $\mathcal{F}$), 
there are $\Omega(\mathcal{F})$ active nodes switching to the
\emph{inactive} state in the second slot of round $t$.
This completes the proof.
\end{proof}

With the above Lemma~\ref{lem:safe}, now the proof idea becomes clear, 
we only need to show that once initiated, the network
will fall into the safe range very soon, and then stays in this range as long as there are
enough active nodes, i.e\ $|A_t| \geq \mathcal{F} \cdot \log n$.
At the very beginning of the protocol, nodes are initiated with constant transmission probabilities, i.e.\
$p_0(\cdot) = \zeta$. Therefore, the summation of the initial transmission probabilities 
might be as large as $n \cdot \zeta$. We need to consider how long it takes for the 
summation $\sum_{v\in A_t} p_t(v)$ to drop below $\mathcal{F}\cdot \alpha_2$, 
where $\alpha_2 > 0$ is the constant defined by the safe range.

\begin{lemma}\label{lem:p_decrease}
For a round $t$ with $\sum_{v\in A_t} p_t(v) = \alpha \cdot \mathcal{F}$, 
it holds that $Pr[\sum_{v\in A_t} p_{t+1}(v) \leq \alpha \cdot \mathcal{F} \cdot 3/4] \geq 7/8$ for large enough $\alpha$.
\end{lemma}
\begin{proof}
To show the conclusion, we need to look at some execution details of the UIE Protocol. 
Note that there are two parts concerning randomness. One part is in the channel selection, 
and the other part is in the transmission selection. Consider the
channel selection part first,
in which a random instance $\sigma$ is a mapping from the $|A_t|$ active nodes to the $\mathcal{F}$ channels. 
Recall that the probability of successful transmission on a channel is
closely related to the total transmission probability of nodes
selecting this channel.
We call an instance \emph{fair} if under it there are at least 
least $\mathcal{F} \cdot 15/16$ channels such that on each of them 
the total transmission probability (of nodes selecting this channel)
is at least $\alpha \cdot 15/16$. By Lemma~\ref{lem:bin_ball}, 
a fraction of $1 - \exp\{-\Omega(\mathcal{F})\}$ of instances are fair.
We next consider such a fair instance $\sigma$.


Let $X_\sigma$ be the random variable that indicates the value of $\sum_{v\in A_{t+1}} p_{t+1}(v)$, 
conditioned on channel selection instance $\sigma$. 
Clearly, $X_{\sigma} \leq 2\sum_{v \in A_t} p_t(v)$, and for different instances $\sigma$, $X_\sigma$s are mutually independent.
For a channel $c$, if without confusion, we also use $c$ to denote the set of active nodes selecting channel $c$ in the instance $\sigma$. Denote by $X^c_\sigma$ the random variable that indicates the value of $\sum_{v\in c\cap A_{t+1}} p_{t+1}(v)$. Hence, $X_\sigma=\sum_{c}X^c_\sigma$.

Focus on a channel $c$ with $\sum_{v\in c\cap A_t} p_t(v) \geq 15\alpha/16$.
The probability that there is at least one transmission on channel $c$ is 
at least $1 - \exp\{-15\alpha/16\}$, by Equation~(\ref{eqn:idle}).
According to the UIE Protocol, the nodes that selected channel $c$
all halve their transmission probabilities if channel $c$ is not idle in round $t$. Hence, 
\begin{eqnarray*}
&&Pr[X^c_\sigma = \sum_{v\in c\cap A_t} \frac{p_t(v)}{2}| \sum_{v\in c\cap A_t} p_t(v) \geq \alpha\cdot \frac{15}{16}] \\
&\geq& 1 - \exp\{-\alpha \cdot \frac{15}{16}\},
\end{eqnarray*}
which is at least $31/32$ when $\alpha$ is large enough. 
Hence in expectation, there are at least $(31/32)\cdot (\mathcal{F} \cdot 15/16)$ channels $c$
with $X^c_\sigma = \sum_{v\in c\cap A_t} p_t(v)/2$. 

Note that once the instance $\sigma$ is given, the total transmission
probability $\sum_{v\in c\cap A_t} p_t(v)$ for each channel $c$ is
specified. Then for different channels, the random variables
$X^c_\sigma$s are mutually independent.
Hence, by the Chernoff bound in Lemma~\ref{lem:chernoff}, there are at least $(15/16)\cdot (\mathcal{F} \cdot15/16)$ channels
with $X^c_\sigma = \sum_{v\in c\cap A_t} p_t(v)/2$ with probability $1 - \exp\{-\Omega(\mathcal{F})\}$. 
Hence, with probability $1 - \exp\{-\Omega(\mathcal{F})\}$, 
\begin{equation*}
\begin{aligned}
X_\sigma &\leq \mathcal{F} \cdot \left(\frac{15}{16}\right)^2
\cdot\frac{15\alpha}{16}\cdot \frac{1}{2} + (\alpha \cdot \mathcal{F} - \alpha \cdot \mathcal{F} \cdot \left(\frac{15}{16}\right)^3) \cdot 2\\ 
&< \alpha \cdot \mathcal{F} \cdot 3/4.
\end{aligned}
\end{equation*}
Finally it holds that 
$Pr[\sum_{v\in A_{t+1}} p_{t+1}(v) \leq \alpha \cdot \mathcal{F} \cdot 3/4] 
\geq (1 - \exp\{-\Omega(\mathcal{F})\}) \cdot (1 - \exp\{-\Omega(\mathcal{F})\})$ 
which is at least $7/8$ for large $\mathcal{F}$. 
The last thing to note is in the above analysis we did
not consider the effect when an active node becomes \emph{inactive},
which only makes the summation decrease and hence is not harmful. 
\end{proof}

\begin{lemma}[Going down]\label{lem:most_down}
There exists a constant $\alpha'_2 > 1$, such that among $\gamma \log n$ rounds (not necessarily consecutive)
with $\sum_{v\in A_t} p_t(v) \geq \alpha'_2 \cdot \mathcal{F}$ and sufficiently large $\gamma > 0$,
there are at least $\frac{3}{4}\gamma \log n$ rounds
with $\sum_{v \in A_{t+1}} p_{t+1}(v) < \frac{3}{4}\sum_{v\in A_t} p_t(v)$, with probability $1 - O(n^{-1})$.
\end{lemma}
\begin{proof}
Let $T := \gamma \log n$, and $X_t$ be the random variable that indicates
the value of $\sum_{v\in A_{t+1}} p_{t+1}(v) / \sum_{v\in A_t} p_t(v)$. 
Then by Lemma~\ref{lem:p_decrease}, it holds that $Pr[X_t \leq 3/4] \geq 7/8$.
Let $Y_t$ be the binary random variable that takes value $1$ if $X_t \leq 3/4$.   
Note that given $\sum_{v\in A_t} p_t(v) > \alpha'_2 \cdot \mathcal{F}$, $\mathbb{E}[Y_t] \geq 7/8$ always hold.
Hence, $\mathbb{E}[\sum_{t=1}^T Y_t] \geq T\cdot 7/8$, and it holds that 
$Pr[\sum_{t=1}^T Y_t \leq T\cdot 3/4] = O(n^{-1})$ by the Chernoff bound. 
\ignore{(\li{Here we actually
use the Chernoff bound with condition $\mathbb{E}[\Pi^T_{t=1} Y_t] \geq \Pi^T_{t=1} 7/8$})
}
That is, with probability $1 - O(n^{-1})$, there are at least $T \cdot 3/4$ rounds $t$ with 
$\sum_{v\in A_{t+1}} p_{t+1}(v) / \sum_{v\in A_t} p_t(v) \leq 3/4$, which completes the proof.
\end{proof}

\begin{lemma}[Fast adaptation]\label{lem:adaption}
There exists a constant $\alpha'_2 > 1$, such that during any period of $\gamma \log n$ rounds
with sufficiently large $\gamma > 0$, the probability that within the considered
period there is a round $t$ with $\sum_{v\in A_t} p_t(v) \leq \alpha'_2 \cdot \mathcal{F}$ is $1 - O(n^{-1})$.
\end{lemma}
\begin{proof}
Denote $T := \gamma \log n$. 
Without loss of generality, assume that the period of $T$ rounds starts from $t = 1$ and ends at $t = T$, 
with $\sum_{v\in A_t} p_t(v) > \alpha'_2 \cdot \mathcal{F}$ always holds.
Note that 
\[\sum_{v\in A_T} p_T(v) = \sum_{v \in A_0} p_0(v) 
\cdot \Pi_{t=0}^{T-1} \frac{\sum_{v\in A_{t+1}} p_{t+1}(v)}{ \sum_{v\in A_t} p_t(v)}.\]
Then by Lemma~\ref{lem:most_down}, with probability at least $1 - O(n^{-1})$, it holds that 
\begin{equation}
\begin{aligned}
\sum_{v \in A_{T}} p_T(v) &\geq \sum_{v\in A_0} p_0(v) \cdot \left( \frac{3}{4} \cdot \frac{3}{4} 
\cdot \frac{3}{4} \cdot 2\right)^{\frac{T}{4}}\\
&=\sum_{v\in A_0} p_0(v) \cdot \left(\frac{27}{32}\right)^{\frac{T}{4}},
\end{aligned}
\end{equation}
where the first inequality holds by 
by coupling the ``evolution'' factors $\sum_{v\in A_{t+1}} p_{t+1}(v) / \sum_{v\in A_t} p_t(v)$. 
Since it holds that $\sum_{v\in A_0} p_0(v) < n$ 
and $T = \gamma \log n$, we know that $\sum_{v\in A_T} p_T(v)$ is at most $\alpha'_2\cdot \mathcal{F}$ for large enough $\gamma$.
\end{proof}

In the above, we have shown that the adaptation process of the total
transmission probability (from the initial state to
$\Theta(\mathcal{F})$) takes $O(\log n)$ rounds with high probability.
Meanwhile, we also showed that
when the total transmission probability increases beyond the upper bound
of the safe range, the total transmission probability of
active nodes shows a trend of going down. To finally show that in most of the
rounds, the safe range is satisfied, we still need to show that if the
total transmission probability of active nodes becomes very small, the
trend is that it will go up.


\begin{lemma}\label{lem:p_increase}
There exists $\alpha' \geq 0.01$ such that for any time $t$ with $\sum_{v\in A_t} p_t(v) = \alpha' \cdot \mathcal{F}$, 
it holds that 
\[Pr[\sum_{v\in A_{t+1}} p_{t+1}(v) \geq \alpha' \cdot \mathcal{F} \cdot \frac{4}{3}] \geq \frac{7}{8}.\]
\end{lemma}
\begin{proof}
By Lemma~\ref{lem:bin_ball}, there is a fraction of $1 - \exp\{-\Omega(\mathcal{F})\}$ of 
channel selection instances in which on at least $\mathcal{F} \cdot 15/16$ channels the total transmission probability of active nodes selecting the channel is between $\alpha' \cdot 15/16$ and $\alpha' \cdot 2$ (i.e.\ the \emph{fair} instances). 
Consider such an instance $\sigma$.

Let $X_\sigma$ be the random variable that indicates the value of $\sum_{v\in A_{t+1}} p_{t+1}(v)$, 
conditioned on channel selection $\sigma$. Clearly, $X_\sigma \leq \sum_{v\in A_t} 2\cdot p_t(v)$, and 
$X_\sigma$'s with distinct $\sigma$'s are independent. For a channel $c$, if without confusion, we also use $c$ to denote the set of active nodes selecting $c$. Let $X^c_\sigma$ denote the random
variable that indicates the value of $\sum_{v \in c\cap A_{t+1}} p_{t+1}(v)$. Hence, $X_\sigma=\sum_{c}X^c_\sigma$.
We consider a channel $c$ with $\alpha'\cdot 15/16 \leq \sum_{v\in c\cap A_t} p_t(v) \leq \alpha' \cdot 2$.
The probability that there are no transmissions on channel $c$ is at least $4^{-2\alpha'}$. 
Hence $Pr[X^c_\sigma = \sum_{v \in c\cap A_t} 2\cdot p_t(v) | \alpha'\cdot 15/16 
\leq \sum_{v\in c\cap A_t} p_t(v)\leq \alpha'\cdot 2] \geq 4^{-2\alpha'}$ 
which is at least $31/32$ when $\alpha'$ is close to $0.01$. 
Hence, in expectation there are at least $(31/32)\cdot (\mathcal{F} \cdot 15/16)$ channels $c$ with
$X^c_\sigma = \sum_{v\in c\cap A_t} 2 \cdot p_t(v)$.

Note that once the instance $\sigma$ is given, the total transmission probability of nodes $\sum_{v\in c\cap A_t} p_t(v)$ 
on each channel $c$ is specified. Then for different channels, the
random variables $X^c_\sigma$s are mutually independent. Hence, by the Chernoff bound, with probability 
$1 - \exp\{-\Omega(\mathcal{F})\}$ there are at least $(15/16)\cdot (\mathcal{F} \cdot 15/16)$ channels $c$ 
with $X^c_\sigma = \sum_{v\in c\cap A_t} 2 \cdot p_t(v)$. Hence, with probability $1 - \exp\{-\Omega(\mathcal{F})\}$,
$X_\sigma \geq (\mathcal{F} \cdot 15^2/16^2)\cdot (15\alpha'/16) \cdot 2 + (\alpha'\cdot \mathcal{F} 
- \alpha'\cdot \mathcal{F} \cdot 15^3/16^3) /2 - \zeta \cdot \mathcal{F} \cdot 17/16$, 
where the loss of weight $\zeta \cdot \mathcal{F} \cdot 17/16$ is due
to those active nodes switching to the \emph{inactive} state: at most $\mathcal{F} / 16$ active nodes
become \emph{inactive} in the second slot, and at most $1\leq \mathcal{F}$ active nodes become \emph{inactive} 
in the $4$th slot. When $\zeta$ is small enough, $\zeta \cdot \mathcal{F} \cdot 17/16$ is very small 
compared to $\mathcal{F} \cdot \alpha'$, and hence $X_\sigma \geq \mathcal{F} \cdot \alpha' \cdot 4/3$ (with probability 
$1 - \exp\{\Omega(\mathcal{F})\}$).
Finally it holds that 
$Pr[\sum_{v\in A_{t+1}} p_{t+1}(v) \geq \alpha'\cdot \mathcal{F} \cdot 4/3]
\geq (1-\exp\{-\Omega(\mathcal{F})\})\cdot (1-\exp\{-\Omega(\mathcal{F})\})$
which is at least $7/8$ for large $\mathcal{F}$.
\end{proof}

\begin{lemma}[Going up]\label{lem:most_up}
There exists a constant $\alpha'_1 > 0$, such that among $\gamma \log n$ 
rounds (not necessarily consecutive)
with $\sum_{v\in A_t} p_t(v) \leq \alpha'_1 \cdot \mathcal{F}$ and sufficiently large $\gamma > 0$, 
there are at least $\frac{3}{4}\gamma\log n$ rounds
with $\sum_{v\in A_{t+1}} p_{t+1}(v) \geq \frac{4}{3}\sum_{v\in A_t} p_t(v)$, with probability $1 - O(n^{-1})$.
\end{lemma}
\begin{proof}
Let $T := \gamma \log n$, and $X_t$ be the random variable that indicates the value of 
$\sum_{v\in A_{t+1}} p_{t+1}(v) / \sum_{v\in A_t} p_t(v)$. Then by Lemma~\ref{lem:p_increase}, 
it holds that $Pr[X_t \geq 4/3] \geq 7/8$.\\
\indent Let $Y_t$ be the binary random variable that takes value $1$ if $X_t \geq 2$. Note that
given $\sum_v p_t(v) < \alpha'_1 \cdot \mathcal{F}$, $\mathbb{E}[Y_t]
\geq 7/8$ always holds.
Hence, $\mathbb{E}[\sum^T_{t=1} Y_t] \geq T\cdot 7/8$, and it holds that $Pr[\sum^T_{t=1} Y_t \leq T\cdot 3/4] = O(n^{-1})$
by the Chernoff bound. That is, with probability $1 - O(n^{-1})$, there are at least 
$T \cdot 3/4$ rounds $t$ with $\sum_{v\in A_{t+1}} p_{t+1}(v) / \sum_{v\in A_t} p_t(v) \geq 4/3$.
\end{proof}

Now we are ready to show that in most of the rounds after the adaptation
process, the total transmission probability of active nodes is in the
safe range.

\begin{lemma}[Stable]\label{lem:stable}
Let $t_0$ be the first round in which $\sum_{v\in A_t} p_{t_0}(v)$ drops below $\alpha'_2 \cdot \mathcal{F}$. 
In the subsequent $T:=\tau \cdot \log n$ rounds where $\tau > 0$ and $n$ are large enough,
the following hold:

$(i)$ \textbf{hardly going high:} there are at least $T \cdot 3/4$ rounds $t$ with
$\sum_{v \in A_t} p_{t}(v) \leq \alpha_2 \cdot \mathcal{F}$, where $\alpha_2 > \alpha'_2$ is a constant.

$(ii)$ \textbf{hardly going low:} there are at least $T \cdot 3/4$ rounds $t$ with
$\sum_{v\in A_t} p_{t}(v) \geq \alpha_1 \cdot k$, where $\alpha_1 < \alpha'_1$ is a constant.

\end{lemma}
\begin{proof} 
We prove the two conclusions one by one.

\textbf{Proof for ``hardly going high''.}
Consider the period from $t = t_0$ to $t = t_0+T$. Define a \emph{wave} to be 
an interval $[t_1, t_2]$ with $t_2 > t_1 + 19$, such that for rounds $t \in[t_1,t_2]$ 
it holds that $\sum_v p_t(v) > \alpha'_2 \cdot \mathcal{F}$, and for rounds $t = t_1 - 1, t_2 + 1$ it holds
that $\sum_{v\in A_t} p_t(v) \leq \alpha'_2 \cdot \mathcal{F}$. Then for any round $t$ not in a wave, 
$\sum_{v\in A_t} p_t(v)$ is at most $\alpha_2 \cdot \mathcal{F}$ where $\alpha_2 := \alpha'_2 \cdot 2^{10}$.

Assume there are at least $T\cdot 1/4$ rounds $t$ with $\sum_{v \in A_t} p_t(v) > \alpha_2 \cdot \mathcal{F}$.
Otherwise, the lemma holds. Let $\mathcal{A}$ denote the event that the assumption is true. 
Next, we show that $\mathcal{A}$ will never happen when $n$ is large enough. 
Let $T'$ be the number of rounds $t$ with $\sum_{v\in A_t} p_t(v) > \alpha_2 \cdot \mathcal{F}$.
Clearly, these rounds are all on waves, and by the assumption, $T^{'}\geq \frac{1}{4}T$.
Let $\mathcal{B}$ denote the event that among all these rounds, 
there are $T'\cdot 3/4$ rounds $t$ with $X_t \leq 3/4$. Recall that
$X_t$ is the random variable that takes value $\sum_{v\in A_{t+1}} p_{t+1}(v) / \sum_{v\in A_t} p_t(v)$. 
Assume that $\tau > 4 \gamma$, where $\gamma$ is from Lemma~\ref{lem:most_down}.
Then $T' > \gamma\log n$, and hence by Lemma~\ref{lem:most_down}, 
it holds that $Pr[\mathcal{B}|\mathcal{A}] = 1 - O(n^{-1})$, which is positive when $n$ is large enough. 
However, as shown in the following argument, events $\mathcal{B}$ and
$\mathcal{A}$ do not happen together, which leads to the conclusion
that $\mathcal{A}$ will never happen when $n$ is large enough.

Now we show that $\mathcal{B}$ and $\mathcal{A}$ do not happen together. Actually, it is sufficient
to show that $\mathcal{B}$ will not happen. Recall that event $\mathcal{B}$ happens meaning that
a fraction of $3/4$ rounds in waves satisfy $X_t \leq 3/4$. To show this is impossible, 
we focus on a single wave $[t_1, t_2]$, and prove that among these $t_2 - t_1 + 1$ rounds, 
there are less than $(t_2 - t_1 + 1) \cdot 3/4$ rounds $t$ with $X_t \leq 3/4$.
Assume the opposite, and then the value of $\sum_{v\in A_{t_2}} p_{t_2}(v)$ is at most $\sum_{v\in A_{t_1}} p_{t_1}(v) 
\cdot \left(27/32\right)^{(t_2 - t_1 + 1)/4}$ (using the coupling technique). Recalling that
in a wave $t_2 - t_1 + 1 > 20$, we have $\sum_{v\in A_{t_2}} p_{t_2}(v) < \sum_{v\in A_{t_1}} p_{t_1}(v) 
\cdot \left(27/32\right)^5 < \sum_{v\in A_{t_1}} p_{t_1}(v)/2$. Since in round $t = t_1 - 1$, 
$\sum_{v\in A_t} p_t(v) < \alpha'_2 \cdot \mathcal{F}$,
which implies that $\sum_{v\in A_{t+1}} p_{t_1}(v) \leq 2 \alpha'_2 \cdot \mathcal{F}$. Hence, 
$\sum_{v\in A_{t_2}} p_{t_2}(v) < \alpha'_2 \cdot \mathcal{F}$,
which contradicts the definition of the wave. Hence the assumption does
not hold, which completes the proof.

\textbf{Proof for ``hardly going low''.}
Consider the period from $t = t_0$ to $t = t_0+T$. Define a \emph{hole} to be an 
interval $[t_1, t_2]$ with $t_2 > t_1 + 19$, such that for rounds $t = t_1, \ldots, t_2$ it holds
that $\sum_v p_t(v) < \alpha'_1 \cdot \mathcal{F}$, and for rounds $t = t_1 - 1, t_2 + 1$ it holds that 
$\sum_v p_t(v) \geq \alpha'_1 \cdot \mathcal{F}$. Then for any round $t$ not in a hole, $\sum_v p_t(v)$ 
is at least $\alpha_1 \cdot k$ where $\alpha_1 := \alpha'_1 / 2^{10}$.

Assume there are at least $T \cdot 1/4$ rounds $t$ with $\sum_{v\in A_t} p_t(v) < \alpha_1 \cdot \mathcal{F}$.
Otherwise, the lemma holds. Let $\mathcal{A}$ denote the event that the assumption is true. Next, we show
that $\mathcal{A}$ will never happen when $n$ is large enough.
Let $T'$ be the number of rounds $t$ with $\sum_v p_t(v) < \alpha_1 \cdot \mathcal{F}$. 
Clearly, these rounds are all in holes, and by the assumption, $T'\geq\frac{1}{4}T$. Let $\mathcal{B}$ denote the event that among all these rounds,
there are $T'\cdot 3/4$ rounds $t$ with $X_t \geq 4/3$. Recall that $X_t$ is 
the random variable that takes value $\sum_{v\in A_{t+1}} p_{t+1}(v) / \sum_{v\in A_t} p_t(v)$.
Assume that $\tau > 4\gamma$, where $\gamma$ is from Lemma~\ref{lem:most_up}. 
Then $T' > \gamma \log n$, and hence by Lemma~\ref{lem:most_up}, it holds 
$Pr[\mathcal{B} | \mathcal{A}] = 1 - O(n^{-1})$, which is positive when $n$ is large enough. However, 
as shown in the following argument, $\mathcal{B}$ and $\mathcal{A}$ do
not happen together, which leads to the conclusion
that $\mathcal{A}$ will never happen when $n$ is large enough. 

Now we show that $\mathcal{B}$ and $\mathcal{A}$ never happen together. Actually, it is
sufficient if we show $\mathcal{B}$ never happen. Recall that if event $\mathcal{B}$ happens, it means a fraction
of $3/4$ of the considered rounds satisfy $X_t \geq 4/3$. To show this
is impossible, we focus on a single
hole $[t_1, t_2]$, and prove that among these $t_2 - t_1 + 1$ rounds, there are less than 
$(t_2 - t_1 + 1)\cdot 3/4$ rounds $t$ with $X_t \geq 4/3$.
Assume the opposite, and then the value of $\sum_{v\in A_{t_2}} p_{t_2}(v)$ is at least 
$\sum_{v\in A_{t_1}} p_{t_1}(v)\cdot (32/27)^{(t_2 - t_1 + 1)/4}$ (using the coupling technique). 
Recalling that in a hole $t_2 - t_1 + 1 > 20$, we have 
$\sum_{v\in A_{t_2}} p_{t_2}(v) > \sum_{v\in A_{t_1}} p_{t_1}(v) \cdot (32/27)^5 > \sum_{v\in A_{t_1}} 2\cdot p_{t_1}(v)$. 
Since at $t = t_1 - 1$, $\sum_{v\in A_t} p_t(v) > \alpha'_1 \cdot \mathcal{F}$, 
which implies that $\sum_{v\in A_{t_1}} p_{t_1}(v) \geq \alpha'_1 \cdot \mathcal{F}/2$.
Hence, $\sum_{v\in A_{t_2}} p_{t_2}(v) > \alpha'_1 \cdot \mathcal{F}$, 
which contradicts the definition of the hole. Hence, the hypothesis does
not hold, which completes the proof.
\end{proof}

Now, we are ready to prove Lemma~\ref{lem:multi}. 

\textbf{\emph{Proof of} Lemma~\ref{lem:multi}}. At first, recall that 
by Lemma~\ref{lem:safe}, in any round $t$ with $|A_t| \geq \mathcal{F} \cdot \log n$
and $\alpha_1 \cdot \mathcal{F} \leq \sum_{v\in A_t} p_t(v) \leq \alpha_2 \cdot \mathcal{F}$, there exists constants $0 < c_1, c_2 < 1$ such that 
with probability at least $c_1$ there are $c_2 \cdot \mathcal{F}$ active nodes switching to the 
\emph{inactive} state. 

Define $T_1$ as the first round $t$  
such that the summation $\sum_{v \in A_t} p_t(v)$ drops below $\alpha_2 \cdot k$. 
By Lemma~\ref{lem:adaption}, we know that $T_1 = O(\log n)$. After $T_1$, by applying 
Lemma~\ref{lem:stable} it follows that for any period of length at least
$T' := \max\{2\cdot k/ (\mathcal{F} \cdot c_1 \cdot c_2), \tau \cdot \log n\}$, 
with high probability, there are $T'/2$ rounds $t$ in which $\sum_{v \in A_t} p_t(v)$ 
is between $\alpha_1 \cdot \mathcal{F}$ and $\alpha_2 \cdot \mathcal{F}$. 
Then we know that for large enough $\tau > 0$, with high probability there is 
a round $t < T_1 + T'$ that satisfies $|A_t| < \mathcal{F} \cdot \log n$.
Otherwise, based on the above argument and using the Chernoff bound, it is easy to show that up to round $T_1 + T'$, there 
are more than $k$ active nodes switching to the \emph{inactive} state with high
probability, which is impossible.

Hence, there exists constant $\gamma' > 0$ with $T := \gamma'(\log n + k/\mathcal{F}) \geq T_1 + T'$, 
such that with high probability there is a round $t\leq T$ that satisfies 
$|A_t| < \mathcal{F} \cdot \log n$. Recall that we assume $k \geq \mathcal{F}\cdot \log n$ (otherwise,
we can ignore this section and only consider the analysis in Section~\ref{sec:primary}), 
which implies $T = O(k/\mathcal{F})$. \qed

\subsection{Efficiency of the Primary Channel}\label{sec:primary}

In this section, we analyze the ``second part'' of the algorithm execution: the execution after the round when the number of active nodes drops
below $\mathcal{F} \cdot \log n$.  The conclusion is summarized in
Lemma~\ref{lem:primary}. Note that here in this part of the analysis, we
do not consider the decrease of active nodes due to successful
transmissions in the multi-channel transmission process. Since the
multi-channel transmission process makes the decrease of active nodes
much faster, the assumption will not affect the correctness of the
analysis.

\begin{lemma}\label{lem:primary}
Consider a round $T$ with $|A_T| \leq \mathcal{F} \cdot \log n$. There is
a constant $\mu > 0$ such that at time $T^* \leq T + \mu \cdot \mathcal{F} \cdot \log n$
there is only one active node left with high probability.
\end{lemma}
\begin{proof}
The proof for this lemma depends on a special case of the proof for Lemma~\ref{lem:multi}, 
where $\mathcal{F} = 1$ and the transmission probability refers to
$q(\cdot)$. Hence, we only give a brief sketch.

After time $T$ with $|A_T| \leq \mathcal{F} \cdot \log n$, it takes at most $O(\log n)$
rounds for the summation $\sum_{v \in A_t} q_t(v)$ to fall down 
to a range between $\beta_1$ and $\beta_2$. Here,
$\beta_1$ and $\beta_2$ are constants such that for any round $t$
with $\beta_1 \leq \sum_{v\in A_t} q_t(v) \leq \beta_2$, 
there is one active node switching to the \emph{inactive} state in the $4$th slot with constant probability.
Afterward, consider a round $T' := T + \mu \cdot \mathcal{F} \cdot \log n$ where 
$\mu > 0$ is a large enough constant. Then with high probability there is 
a time round $t < T'$ such that $|A_t| = 1$. Otherwise, during the period 
from $T$ to $T'$, with high probability there are more than $\mathcal{F}\cdot \log n$ active nodes switching 
to the \emph{inactive} state in the $4$th slot, which is impossible.
\end{proof}

\subsection{Stabilization} 

Recall in Lemma~\ref{lem:adaption}, we have proved that it takes
$O(\log n)$ rounds for a network to become ``safe'', which means
the summation $\sum_{v\in A_t} p_t(v)$ goes from its initial value to
a range between $\alpha_1\cdot \mathcal{F}$ and $\alpha_2 \cdot \mathcal{F}$ for some constants
$\alpha_1, \alpha_2 > 0$. This conclusion can be generalized to any network 
state that is not ``safe''. We describe the generalized conclusion formally in the following
Theorem~\ref{thm:stablize}.

\begin{theorem}\label{thm:stablize}
Consider the case when the number of \emph{active} nodes is always at least $\mathcal{F}\cdot \log n$.
For a round $t^*$ with $\sum_{v\in A_{t^*}} p_{t^*}(v)$ outside the safe range 
$[\alpha_1 \cdot \mathcal{F}, \alpha_2 \cdot \mathcal{F}]$, 
with high probability $\sum_{v\in A_{t}} p_{t}(v)$ will fall into the safe range in $\Phi=O(\log(\max\{\frac{p^{*}}{\mathcal{\mathcal{F}}},\frac{\mathcal{\mathcal{F}}}{p^{*}}\})+\log n)$ rounds.
\end{theorem}
\begin{proof}
Recall that in the proof of Lemma~\ref{lem:adaption}, in order to show that 
the summation $\sum_{v\in A_t} p_t(v)$ goes below $\alpha_2\cdot \mathcal{F}$, we considered 
$T:= 4\cdot T'$ rounds such that $T'\geq \gamma\log n$ for a large enough constant $\gamma$, during which there are $3\cdot T'$ rounds with a decrease
of $\sum_{v\in A_t} p_t(v)$ by a factor $3/4$ (by Lemma~\ref{lem:most_down}) and $T'$ rounds with an increase of $\sum_{v\in A_t} p_t(v)$ by a factor at most $2$. 
Then, after these $T$ rounds, the summation $\sum_{v\in A_t} p_t(v)$ 
will be decreased by a factor of $(27/32)^{T'}$ with high probability. 
Since the network is initiated with $\sum_{v\in A_0} p_0(v) \leq \zeta \cdot n$, we know that it is enough to 
set $T := O(\log n)$ for the network to become ``safe''. 

In a similar approach, it is easy to show that for any round $t^*$ with $p^* := \sum_{v\in A_{t^*}} p_{t^*}(v) > \alpha_2 \cdot \mathcal{F}$,
by the round $t' := t^*+\max\{4 \cdot \log (32 \cdot p^*/(27\cdot \alpha_2\cdot \mathcal{F})), 4\gamma\log n\}$, the summation
$\sum_{v\in A_{t'}} p_{t'}(v)$ becomes smaller than $\alpha_2 \cdot \mathcal{F}$ with high probability.

For the case that $p^* < \alpha_1 \cdot \mathcal{F}$, the proof idea is similar. Note that 
during $T:= 4\cdot T'$ rounds with $\sum_{v\in A_t} p_t(v) < \alpha_1 \cdot \mathcal{F}$, where $T'\geq \gamma\log n$ for a large enough constant $\gamma$,  
there are $3\cdot T'$ rounds with an increase of $\sum_{v\in A_t} p_t(v)$ by a factor $4/3$ 
(Lemma~\ref{lem:most_up}) and $T'$ rounds with a decrease of $\sum_{v\in A_t} p_t(v)$ by a factor $1/2$. 
Overall, after these $T$ rounds, the summation $\sum_{v\in A_t} p_t(v)$ 
will be increased by a factor of $(32/27)^{T'}$ with high probability. Hence, by setting $t' := t^* + \max\{4 \cdot \log (27 \cdot \alpha_1\cdot \mathcal{F}/(32\cdot p^*)), 4\gamma\log n\}$, the summation
$\sum_{v\in A_{t'}} p_{t'}(v)$ becomes larger than $\alpha_1 \cdot \mathcal{F}$ by round $t'$ with high probability.
\end{proof}

\section{Conclusion} \label{sec:con}

In this paper, we considered the information exchange problem of $k$ source nodes 
in single-hop multiple-channel networks of $n$ nodes. With $\mathcal{F}$ 
available channels and collision detection, we proposed a protocol that 
solves the information exchange problem in $O(k/\mathcal{F} + \mathcal{F}\cdot \log n)$ rounds, 
with high probability. Our algorithm is uniform in $n$ and $k$, which is
the first known uniform algorithm for information exchange in multi-channel
networks. And the proposed protocol is asymptotically optimal when $k$
is large.

In our protocol, when detecting transmissions, a node will decrease its transmission probability to avoid collisions.
Then if there exist jamming signals on a channel, an analysis
similar to that introduced in this paper would show that even for the case when
jamming only affects a constant fraction of the available channels, the
total transmission probability (i.e.\ $\sum_{v\in A_t} p_t(v)$)
may tend to become very small. The affects the primary channel strategy 
even more significantly, since a fixed channel may be jammed all the time.
This problem 
motivates us to consider jamming-resilience of the proposed protocol
in the future.

\bibliographystyle{abbrv}
\bibliography{sigproc}  
\ignore{
\appendix

\section{Headings in Appendices}
}
\end{document}